%% file: Feedback_ISIT17.tex
\documentclass[conference]{IEEEtran}
\onecolumn
\input{MATH_template.tex}

\begin{document}

\title{ On the Necessity of Structured Codes for Communications over MAC with Feedback}
\author{
 \IEEEauthorblockN{Mohsen Heidari}
  \IEEEauthorblockA{EECS Department\\University of Michigan\\ Ann Arbor,USA \\
    Email: mohsenhd@umich.edu} 

\and
 \IEEEauthorblockN{Farhad Shirani}
  \IEEEauthorblockA{EECS Department\\University of Michigan\\ Ann Arbor,USA \\
    Email: fshirani@umich.edu }

  \and
  \IEEEauthorblockN{S. Sandeep Pradhan}
  \IEEEauthorblockA{EECS Department\\University of Michigan\\ Ann Arbor,USA \\
    Email: pradhanv@umich.edu}
}
\IEEEoverridecommandlockouts

\maketitle
\begin{abstract}
The problem of three-user multiple-access channel (MAC) with noiseless feedback is investigated. A new coding strategy is presented. The coding scheme builds upon the natural extension of the Cover-Leung (CL) scheme \cite{Cover-Leung}; and uses quasi-linear codes. A new single-letter achievable rate region is derived.
The new achievable region strictly contains the CL region. This is shown through an example. In this example, the coding scheme achieves optimality in terms of transmission rates. It is shown that any optimality achieving scheme for this example must have a specific algebraic structure. Particularly,  the codebooks must be closed under binary addition.   
\end{abstract}

\section{Introduction}
\IEEEPARstart{T}{he} problem of three user MAC with noiseless feedback is depicted in Figure \ref{fig: MAC with FB}. This communication channel consists of one receiver and multiple transmitters. After each channel use, the output of the channel is received at each transmitter noiselessly. Gaarder and Wolf \cite{Gaarder-Wolf} showed that the capacity region of the MAC can be expanded through the use of the feedback. This was shown in a binary erasure MAC.
Cover and Leung \cite{Cover-Leung} studied the two-user MAC with feedback, and developed a coding strategy using unstructured random codes.  

The main idea behind the CL scheme is to use superposition block-Markov encoding. The scheme operates in two stages. In stage one, the transmitters send the messages with a rate outside of the no-feedback capacity region (i.e. higher rates than what is achievable without feedback).  The transmission rate is taken such that each user can decode the other user's message using feedback. In this stage, the receiver is unable to decode the messages reliably, since the transmission rates are outside the no-feedback capacity region. Hence, the decoder only is able to form a list of ``highly likely" pairs of messages.  In the second stage, the encoders fully cooperate to send the messages (as if they are sent by a centralized transmitter). The receiver decodes the message pair from its initial list. After the initiation block, superposition coding is used to transmit the sequences corresponding to the two stages.  

\begin{figure}[hbtp]
\centering
\includegraphics[scale=0.6]{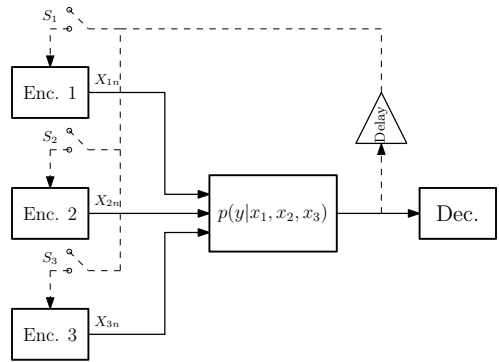}
\caption{The three-user MAC with noiseless feedback. If the switch $S_i$ is closed, the feedback is available at the $i$th encoder, where $i=1,2,3$.}
\label{fig: MAC with FB}
\end{figure}

The single-letter achievable rate region for the CL scheme was characterized in \cite{Cover-Leung}. Later, it was shown that the CL scheme achieves the feedback capacity for a class of MAC with feedback \cite{Willems-FB}. However, this is not the case for the general MAC with feedback \cite{Ozarow}. Several improvements to the CL achievable region were derived \cite{Lapidoth}, \cite{Ramji-Sandeep-FB}. 
In  \cite{Lapidoth} and \cite{Ramji-Sandeep-FB},  additional stages are appended to the CL scheme. In these schemes, the encoders decode each others' messages in several stages.  Kramer \cite{Kramer-thesis}, used the notion of \textit{directed information} to derive the capacity region of the two-user MAC with feedback. However, the characterization is not computable, since it is an infinite letter characterization. Finding a computable characterization of the capacity region remains an open problem.
 
In this work, we study the problem of three-user MAC with feedback. We propose a new coding scheme which builds upon the CL scheme. We derive a computable single-letter achievable rate region for this scheme, and show that the new region improves upon the previous known achievable regions for this problem. Recently, we showed that the application of structured codes results in improved performance for the problem of transmission of sources over the MAC 
  \cite{ISIT-paper-MAC-corr-sources}. Here, we use the ideas proposed in \cite{ISIT-paper-MAC-corr-sources} to prove the necessity of structured codes in the problem of MAC with feedback.  Specifically, we use \textit{quasi-linear} codes that are proposed in \cite{ISIT-paper-quasi-linear}. 

The coding scheme operates in three stages. In stage one, the encoders send independent messages with rates outside of the CL region.  Therefore, encoders are unable to decode each others' messages. However, each encoder can decode the binary sum of the messages of the other two encoders. In stage two, the messages are superimposed on the summation which is decoded in the previous stage. At the end of this stage, the encoders decode each others' messages. Stage three is similar to the second stage in CL scheme. We provide an example where the new coding scheme achieves optimal performance, whereas the previous schemes are suboptimal. Finally, we prove that any optimality achieving coding scheme must use encoders whose set of output sequences is linearly closed. 

The rest of the paper is organized as follows: Section \ref{sec: prelim} presents basic definitions. Section  \ref{sec: capacity} characterizes the capacity region of the three user MAC with feedback. Section  \ref{sec: converse} presents an example of a MAC with feedback, and discusses the necessity of structured codes for that setup. Section \ref{sec: achievable region} contains the main result of the paper, and characterizes a new achievable rate region. Finally, Section \ref{sec: conclusion} concludes the paper.

\section{Preliminaries and Model}\label{sec: prelim}
A three-user discrete memoryless MAC is defined by: 1) three input alphabets $\mathcal{X}_1,\mathcal{X}_2$, and $\mathcal{X}_3$, 2) an output alphabet $\mathcal{Y}$, and 3) a conditional probability distribution $p(y|x_1, x_2, x_3)$ for all $(y, x_1, x_2, x_3)\in \mathcal{Y}\times \mathcal{X}_1\times \mathcal{X}_2 \times \mathcal{X}_3$. Such setup is denoted by $(\mathcal{X}_1,\mathcal{X}_2, \mathcal{X}_3, \mathcal{Y}, P_{Y|X_1X_2X_3})$.  Let $y^n$ be the output sequence corresponding to $n$ uses of the channel, and $x^n_i$ be the input sequence of the channel. Then, following condition is satisfied:
\begin{align}\label{eq: chann probabilities}
p(y_n|y^{n-1}, x_1^{n-1},x_2^{n-1}, x_3^{n-1})=p(y_n|x_{1n}, x_{2n}, x_{3n}).
\end{align}
Figure \ref{fig: MAC with FB} illustrates this setup. In this work, we assume that noiseless feedback is available at a subset of the encoders. 
\begin{definition}
An $(N, M_1, M_2, M_3)$ transmission system for a given three-user MAC with feedback is defined as a sequence of encoding functions and a decoding function. If the feedback is available at the $i$th user, the corresponding encoding functions are defined as
\begin{align*}
f_{i,n}: \{1,2, \dots, M_i\} \times \mathcal{Y}^{n-1}\rightarrow \mathcal{X}_i, 
\end{align*}
where $i=1,2,3$, and $n=1,2,\dots, N$. If the feedback is not available at $i$th encoder, the corresponding encoding functions are defined as 
\begin{align*}
f'_{i,n}: \{1,2, \dots, M_i\} \rightarrow \mathcal{X}_i. 
\end{align*}
The decoding function is defined as 
\begin{align*}
g: \mathcal{Y}^N\rightarrow \{1,2, \dots, M_1\} \times \{1,2, \dots, M_2\}\times \{1,2, \dots, M_3\}.
\end{align*}

\end{definition}
Let $\Theta_i$ denotes the message for $i$th transmitter, $i=1,2,3$. We assume $\Theta_i$ is drawn randomly and uniformly from $\{1, 2, ..., M_i \}$. Furthermore, we assume $\Theta_i, \Theta_i$, and $\Theta_i$ are mutually independent. The average probability of error for this setup is  
\begin{align*}
\bar{P}=\frac{1}{M_1M_2M_3} \sum_{\theta_1, \theta_2, \theta_3} p(g(\mathbf{Y}^N)\neq (\theta_1, \theta_2, \theta_3)| \theta_1, \theta_2, \theta_3).
\end{align*} 

\begin{definition}
A rate triple $(R_1,R_2,R_3)$ is said to be achievable for a given MAC with feedback, if for any $\epsilon>0$ there exists an  $(N, M_1, M_2, M_3)$ transmission system such that 
\begin{align*}
\bar{P}<\epsilon, \quad \frac{1}{n}\log_2 M_i \geq R_i-\epsilon, \quad i=1,2,3.
\end{align*} 
\end{definition}
The capacity region of the MAC with feedback is the closure of the set of all achievable rate pairs $(R_1,R_2,R_3)$.


\section{Capacity Region of Three-user MAC with Feedback}\label{sec: capacity}
We extend the results of Kramer for the three-user MAC with feedback. We derive a multi-letter characterization for the capacity region. We use the notion of \textit{directed information} presented in \cite{Kramer-thesis}. The entropy of a random sequence $\mathbf{Y}^n$ causally conditioned on $\mathbf{X}^n$ is defined as 
\begin{align*}
H(\mathbf{Y}^n||\mathbf{X}^n)=\sum_{k=1}^nH(\mathbf{Y}_k|\mathbf{Y}^{k-1}, \mathbf{X}^k). 
\end{align*}
 Directed information from a sequence $\mathbf{X}^n$ to a sequence $\mathbf{Y}^n$ is defined as 
\begin{align*}
I(\mathbf{X}^n \rightarrow \mathbf{Y}^n)&=H(\mathbf{Y}^n)-H(\mathbf{Y}^n||\mathbf{X}^n).
\end{align*}
Directed information from a sequence $\mathbf{X}^n$ to a sequence $\mathbf{Y}^n$ when causally conditioned on $\mathbf{Z}^n$ is defined by 
\begin{align*}
I(\mathbf{X}^n \rightarrow \mathbf{Y}^n|| \mathbf{Z}^n)&=H(\mathbf{Y}^n||\mathbf{Z}^n)-H(\mathbf{Y}^n||\mathbf{X}^n\mathbf{Z}^n).
\end{align*}
For more convenience, we use the following notation
\begin{align}\label{eq: normalized direct info}
I_n(X\rightarrow Y)=\frac{1}{n}I(\mathbf{X}^n \rightarrow \mathbf{Y}^n).
\end{align}
With the above notation, we are ready to derive the capacity region. 
\begin{definition}\label{def: MAC FB capacity}
Given a positive integer $L$ and a MAC with feedback, define $\mathcal{R}_L$ as the convex hull of the set of all rates $(R_1,R_2,R_3)$ such that,
\begin{align}
R_i &\leq I_L(X_i \rightarrow Y|| X_j X_k)\\
R_i+R_j &\leq I_L(X_i,X_j \rightarrow Y||  X_k)\\\label{eq: capacity sum rate}
R_1+R_2+R_3 &\leq I_L(X_1,X_2,X_3 \rightarrow Y),
\end{align}
holds for all $i, j, k \in \{1,2,3\}, i\neq j\neq k$, where  the conditional distribution $p(x_{1,l},x_{2,l}, x_{1,l}| x_{1}^{l-1},x_{2}^{l-1}, x_{1}^{l-1} y^{l-1}) $ equals $\prod_{i=1}^3 p(x_{i,l}| x_{1}^{l-1} y^{l-1}). $
\end{definition}

\begin{preposition}\label{prep: capacity region of MAC-FB}
The capacity region of the three-user MAC with feedback is characterized by 
\begin{align*}
\mathcal{C_{FB}}= \bigcup_{L=1}^\infty \mathcal{R}_L
\end{align*}
\end{preposition}
\begin{proof}
The proof is a generalized version of the result given in \cite{Kramer-thesis} and is omitted.
\end{proof}
Note that this is a multi-letter characterization, and is not computable.  
\section{An Example of a MAC with Feedback}\label{sec: converse}
In this section, we show that coding strategies based on structured codes are necessary for the problem of MAC with feedback. We first provide an example of a MAC with feedback. Then, we propose a coding scheme using linear codes, and  show that such coding scheme achieves optimality in terms of achievable rates.

\begin{example}\label{ex: example}
Consider the three-user MAC with feedback problem depicted in Figure \ref{fig: Exp. setup}. In this setup, there is a MAC with three inputs. The $i$th input is denoted by the pair $(X_{i1}, X_{i2})$, where $i=1,2,3$. The output of the channel is denoted by the vector $(Y_1, Y_{21}, Y_{22})$. Noiseless feedback is available only at the third transmitter.  
\begin{figure}[hbtp]
\centering
\includegraphics[scale=0.6]{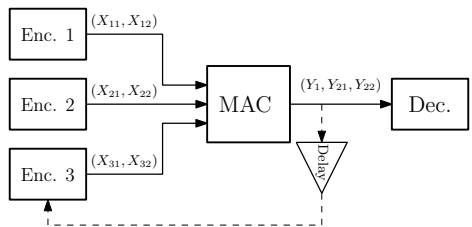}
\caption{The MAC with feedback setup for Example \ref{ex: example}.}
\label{fig: Exp. setup}
\end{figure}
The MAC in this setup consists of two parallel channels. The first channel is a three-user binary additive MAC with inputs $(X_{11}, X_{21}, X_{31})$, and output $Y_1$. The output is related to the inputs via the relation 
\begin{align*}
Y_1=X_{11}\oplus X_{21} \oplus X_{31}\oplus \tilde{N}_\delta,
\end{align*}
where $\tilde{N}_\delta$ is a Bernoulli random variable with bias $\delta$, and is independent of the inputs. 

The second channel is a MAC with $(X_{12}, X_{22}, X_{32})$ as the inputs, and  $(Y_{21},Y_{22})$ as the output. The relation between the output and the input of the channel is depicted in Figure \ref{fig: Exp. second chann}. The channel operates in two states. If the condition $X_{31}=X_{12}\oplus X_{22}$ holds, the channel would be in the first state (the left channel in Figure \ref{fig: Exp. second chann}); otherwise it would be in the second state (the right channel in Figure \ref{fig: Exp. second chann}). In this channel, $N_\delta $ and $N'_\delta$ are Bernoulli random variables with identical bias $\delta$. Whereas, $N_{1/2} $ and $N'_{1/2}$ are Bernoulli random variables with bias $\frac{1}{2}$. We assume that $\tilde{N}_\delta, N_\delta, N'_\delta,N_{1/2} $, and $N'_{1/2}$ are mutually independent, and are independent of all the inputs. 
\begin{figure}[hbtp]
\centering
\includegraphics[scale=0.6]{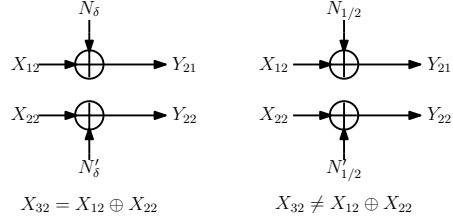}
\caption{The second channel for Example \ref{ex: example}. If the condition $X_{31}=X_{12}\oplus X_{22}$ holds, the channel would be the one on the left; otherwise it would be the right channel.}
\label{fig: Exp. second chann}
\end{figure}
\end{example}

We use linear codes to propose a new coding strategy for the setup given in Example \ref{ex: example}.  The scheme uses a large number $L$ of blocks. The length of each block is $n$. Each encoder has two outputs, one for each channel. We use identical linear codes  with length $n$ and rate $\frac{k}{n}$ for each transmitter.  The coding scheme at each block is performed in two stages. In the first stage, each transmitter encodes the fresh message at the beginning of the block $l$, where $1\leq l \leq L$. The encoding process is performed using the identical linear codes. At the end of the block $l$, the feedback is received by the third user. In stage 2, the third user uses the feedback from the first channel (that is $Y_1$) to decode the binary sum of the messages of the other encoders. Then, it encodes the summation, and sends it through its second output. If the decoding process is successful at the third user, then the relation $X_{32}=X_{12}\oplus X_{22}$ holds with probability one. This is because identical linear codes are used to encode the messages.  As a result of this equality, the channel in Figure \ref{fig: Exp. second chann} is in the first state with probability one. 
In the next Lemma, we show that the rate $$(1-h(\delta), 1-h(\delta), 1-h(\delta))$$ is achievable using this strategy.
%
%
%

\begin{lem}\label{lem: example achievable rate using linear codes}
For the channel given in Example \ref{ex: example}, the rate triple $(1-h(\delta), 1-h(\delta), 1-h(\delta))$ is achievable.
\end{lem} 
\begin{proof}
The proof is given in Appendix \ref{sec: proof of lemma 1 of the example}.
\end{proof}
\begin{remark}\label{rem: optimality}
Based on Preposition \ref{prep: capacity region of MAC-FB}, the triple $(1-h(\delta), 1-h(\delta), 1-h(\delta))$ is a corner point in the capacity region of the channel in Example \ref{ex: example}. This implies the optimality of the above coding strategy in terms of achievable rates. 
\end{remark}

The above coding strategy is different from known schemes in two ways: 1) Identical linear codes are used to encode the messages, 2) The third user uses feedback to decode only the binary sum others' messages.  

\subsection{Converse}
One implication of Remark \ref{rem: optimality} is that the proposed coding scheme achieves optimality.  We show a stronger result in this Subsection. We prove that every coding scheme that achieves  $(1-h(\delta), 1-h(\delta), 1-h(\delta))$, should carry certain algebraic structures such as closeness under the binary addition. 

Suppose there exists a $(N, M_1, M_2, M_3)$ transmission system with rates close to $R_i=1-h(\delta)$, and average probability of error close to $0$, in particular
\begin{align*}
\bar{P}<\epsilon, \quad \frac{1}{n}\log_2 M_i \geq 1-h(\delta)-\epsilon, \quad i=1,2,3,
\end{align*} 
where $\epsilon>0$ is sufficiently small. Since there is no feedback at the first and second encoder, the transmission system predetermines a codebook for user 1 and 2. Note that there are two outputs for encoder 1 and 2. Suppose $\mathcal{C}_{12}$ and $\mathcal{C}_{22}$ are the codebooks assigned to the second output of encoder 1 and encoder 2, respectively. 

Let $\mathbf{X}^N_{i2}$ be the second output of encoder $i$, where $i=1,2,3$.  Let $X_{i2, l}$ denote the $l$th component of $X^N_{i2}$, where  $1\leq l \leq N, ~ i=1,2,3$.  The following lemmas hold for this transmission system. 

\begin{lem}\label{lem: x_2+x_1 needs to be decoded}
 For any fixed $c>0$, define
\begin{align*}
\mathcal{I}_c^N:=\{ l\in [1:N]: P(X_{32, l} \neq X_{12, l}\oplus X_{22, l}) \geq c \}.
\end{align*}
 Then, the inequality $\frac{|\mathcal{I}_c^N|}{N}\leq   \frac{\eta(\epsilon)}{2c(1-h(\delta))}$ holds, where $\eta(\epsilon)$ is a function such that,  $\eta(\epsilon) \rightarrow 0$, as $\epsilon\rightarrow 0$. 
\end{lem}
\begin{proof}
The proof is given in Appendix \ref{sec: proof of lemma 2 of the example}.
\end{proof}
The Lemma implies that in order to achieve $(1-h(\delta), 1-h(\delta), 1-h(\delta))$, the third user needs to decode $ X_{12, l}\oplus X_{22, l}$ for ``almost all" $l\in [1:N]$. This requirement is necessary to insure that the channel given in Figure \ref{fig: Exp. second chann} is in the first state. 

In the next step, we use the results of Lemma \ref{lem: x_2+x_1 needs to be decoded}, and  drive two necessary conditions for decoding $X_{12}\oplus X_{22}$.

\begin{lem}\label{lem: structure in the code}
The following holds 
\begin{align*}
\frac{1}{N}\big| ~\log ||\mathcal{C}_{12}\oplus \mathcal{C}_{22}||- \log ||\mathcal{C}_{12}||~\big| \leq \lambda_1(\epsilon),\\
\frac{1}{N}\big| ~\log ||\mathcal{C}_{12}\oplus \mathcal{C}_{22}||- \log ||\mathcal{C}_{22}||~\big| \leq \lambda_2(\epsilon),
\end{align*}
where $\lambda_j(\epsilon) \rightarrow 0$, as $\epsilon\rightarrow 0, j=1,2$.
\end{lem}
\begin{proof}
The proof is given in Appendix \ref{seq: lem 3}.
\end{proof}
As a result of this lemma,  $\log ||\mathcal{C}_{12}\oplus \mathcal{C}_{22}||$ needs to be close to $\log ||\mathcal{C}_{12}||$ and $\log ||\mathcal{C}_{22}||$. This implies that $\mathcal{C}_{12}$ and $\mathcal{C}_{22}$ possesses an algebraic structure, and are \textit{almost} close under the binary addition. Not that for the case of unstructured random codes $||\mathcal{C}_{12}\oplus \mathcal{C}_{22}||\approx ||\mathcal{C}_{12}||\times ||\mathcal{C}_{22}||$. Hence, unstructured random coding schemes are suboptimal in this example.

\begin{remark}
The three-user extension of CL scheme is suboptimal. Because, the conditions in Lemma \ref{lem: structure in the code} are not satisfied. 
\end{remark}


\section{A New Achievable Rate Region} \label{sec: achievable region} 
In this Section, we use our intuition about the coding scheme in Example \ref{ex: example}, and derive a new computable single-letter achievable rate region for the three-user MAC with feedback problem.   
\begin{definition}
For a given set  $\mathcal{U}$ and a three-user MAC with feedback $(\mathcal{X}_1, \mathcal{X}_2, \mathcal{X}_3, \mathcal{Y}, P_{Y|X_1X_2X_3})$, define $\mathscr{P}$ as the collection of all distributions $P$ of the form 
\begin{align*}
 p(u)p(v_1,v_2,v_3) \prod_{i=1}^3 p(t_i)p(x_i|u,t_i,v_i) p(y|x_1,x_2,x_3),
\end{align*}
for all  $ y\in \mathcal{Y}, u\in \mathcal{U}, t_i\in \FF_2, v_i \in \FF_2, x_i \in \mathcal{X}_i, ~ i=1,2,3$, where  1) $T_1, T_2, T_3$ are mutually independent with uniform distribution over $\FF_2$ , 2) $V_1, V_2, V_3$ are pairwise independent, 3) $p(v_i)=\frac{1}{2}$, and  3) $p(v_1,v_2,v_3)=\frac{1}{4}$.
\end{definition}

Fix a distribution $P\in \mathscr{P} $. Denote $S_i=(X_i, T_i, V_i)$ for $i=1,2,3$. Consider two sets of random variables $(U, S_1,S_2,S_3, Y)$ and $(\tilde{U},\tilde{S}_1, \tilde{S}_2, \tilde{S}_3, \tilde{Y})$. Suppose the distribution of each set of the random variables is $P$.  Then with this notation we have 
\begin{align*}
P_{U S_1 S_2 S_3  Y}=P_{\tilde{U} \tilde{S}_1 \tilde{S}_2\tilde{S}_3 \tilde{Y}}=P
\end{align*}
    
\begin{theorem}\label{thm: MAC-FB achievable}
Consider a MAC $(\mathcal{X}_1, \mathcal{X}_2, \mathcal{X}_3, \mathcal{Y}, P_{Y|X_1X_2X_3})$, and a distribution $P\in \mathscr{P}$. For any subset $\mathcal{A}\subseteq \{1,2,3\}$, and for any  distinct elements $i, j, k \in \{1,2,3\}$ the following bounds hold
\begin{align*}
R_\mathcal{A}&\leq I(X_\mathcal{A};Y|US_{\mathcal{A}^c} \tilde{V}_1 \tilde{V}_2 \tilde{V}_3)+I(U;Y|\tilde{U}\tilde{Y})\\
R_i+R_j &\leq I(T_i\oplus T_j; Y| UT_k X_k \tilde{V}_1 \tilde{V}_2 \tilde{V}_3)\\&+I(\tilde{X}_i \tilde{X}_j;\tilde{Y}|\tilde{U} \tilde{S}_k \tilde{V}_1 \tilde{V}_2 \tilde{V}_3 V_k)\\&+I(\tilde{X}_i \tilde{X}_j;Y | \tilde{U} \tilde{S}_k \tilde{V}_1 \tilde{V}_2 \tilde{V}_3 U S_k \tilde{Y})\\
R_i+R_j &\leq \frac{H(W_i)+H(W_j)}{H(W_i\oplus W_j)} I(T_i\oplus T_j;Y|UT_k X_k),
 \end{align*}
 where 1) $W_i$, is a Bernoulli random variable that is independent of all other random variables, 2) the equality $V_i=\tilde{T}_j\oplus \tilde{T}_k$ holds with probability one, and 3) the Markov chain $$\tilde{U},\tilde{S}_1, \tilde{S}_2, \tilde{S}_3 \leftrightarrow V_1, V_2, V_3 \leftrightarrow U, T_i, X_i, $$ holds for $i=1,2,3$. 
\end{theorem}
\begin{proof}
The proof is given in Appendix \ref{sec: thm 1}.
\end{proof}


\begin{remark}
The rate region in Theorem \ref{thm: MAC-FB achievable} contains the three-user extension of the CL region. For that set $V_1, V_2,V_3$ to be independent of all other random variables. This gives a distribution in $ \mathscr{P}$.
\end{remark}

\section{Conclusion}\label{sec: conclusion}
We proposed a new single-letter achievable rate region for the three-user discrete memoryless MAC with noiseless feedback. 
We used an example to show that this achievable region strictly contains the CL region. In the example, the proposed coding scheme achieves
optimality in terms of transmission rates. Moreover, we proved that any
optimality achieving scheme for this example must have a specific
algebraic structure. Particularly, the codebooks must be closed
under binary addition.
%

\appendices
\section{Proof of Lemma \ref{lem: example achievable rate using linear codes}}\label{sec: proof of lemma 1 of the example}
\begin{IEEEproof}[Outline of the proof]
We start by proposing a coding scheme. There are $L$ blocks of transmissions in this scheme, with new messages available at each user at the beginning of each block. The scheme sends the messages with $n$ uses of the channel. Let $\mathbf{W}^k_{i,[l]}$ denotes the message of the $i$th transmitter at the $l$th block, where $i=1,2,3$, and $1\leq l \leq L$. Let  $\mathbf{W}^k_{i,[l]}$ take values randomly and uniformly from $\FF_2^k$. In this case, the transmission rate of each user is $R_i=\frac{k}{n}, i=1,2,3$. The first and the second outputs of the $i$th encoder in block $l$ is denoted by $\mathbf{X}^n_{i1,[l]}$ and $\mathbf{X}^n_{i2,[l]}$, respectively.

\textbf{Codebook Construction:}
Select a $k\times n$ matrix $\mathbf{G}$ randomly and uniformly from $\FF_2^{k\times n}$. This matrix is used as the generator matrix of a linear code. Each encoder is given the matrix $\mathbf{G}$. Therefore, the encoders use an identical linear code generated by $\mathbf{G}$.

\textbf{Encoder 1 and 2:}
For the first block set $\mathbf{X}^n_{i2,[1]}=0$, for  $i=1,2,3$. For the block $l$, encoder 1 sends $\mathbf{X}^n_{11,[l]}= \mathbf{W}^k_{1,[l]} \mathbf{G}$ through its first output. For the second output, encoder 1 sends  $\mathbf{X}^n_{11,[l-1]}$ from block $l-1$, that is $\mathbf{X}^n_{12,[l]}=\mathbf{X}^n_{11,[l-1]}$. Similarly, the outputs of the second encoder are  $\mathbf{X}^n_{21,[l]}= \mathbf{W}^k_{2,[l]} \mathbf{G}$, and   $\mathbf{X}^n_{22,[l]}=\mathbf{X}^n_{21,[l-1]}$. 
  
\textbf{Encoder 3:}
The third encoder sends  $\mathbf{X}^n_{31,[l]}= \mathbf{W}^k_{3,[l]} \mathbf{G}$ though its first output. This encoder receives  the feedback from the block $l-1$ of the channel. This encoder wishes to decode $\mathbf{W}^k_{1,[l-1]}\oplus\mathbf{W}^k_{2,[l-1]}$ using $\mathbf{Y}^n_{1,[l-1]}$. For this purpose,  it subtracts  $\mathbf{X}^n_{31,[l-1]}$ from $\mathbf{Y}^n_{1,[l-1]}$. Denote the resulting vector by $\mathbf{Z}^n$. Then, it finds a unique vector $\mathbf{\tilde{w}}^k \in \FF_2^k $ such that $(\mathbf{\tilde{w}}^k\mathbf{G}, \mathbf{Z}^n)$ is $\epsilon$-typical with respect to $P_{XZ}$, where $X$ is uniform over $\FF_2$ , and $Z=X\oplus \tilde{N}_\delta$. If the decoding process is successful, the third encoder sends $\mathbf{X}^n_{32,[l]}= \mathbf{\tilde{w}}^k_{[l-1]} \mathbf{G}$. Otherwise, an event $E_{1, [l]}$ is declared.

\textbf{Decoder:}
The decoder receives the outputs of the channel from the $l$th block, that is $\mathbf{Y}^n_{1,[l]}$ and $\mathbf{Y}^n_{2,[l]}$. The decoding is performed in three steps. First, the decoder uses $\mathbf{Y}^n_{2,[l]}$ to decode $ \mathbf{W}^k_{1,[l-1]}$, and $ \mathbf{W}^k_{2,[l-1]}$. In particular, it finds unique $\mathbf{\tilde{w}}_1^k, \mathbf{\tilde{w}}_2^k \in \FF_2^k$ such that $(\mathbf{\tilde{w}}_1^k\mathbf{G}, \mathbf{\tilde{w}}_2^k\mathbf{G}, \mathbf{Y}^n_{2,[l]})$ are jointly $\epsilon$-typical with respect to $P_{X_{12}X_{22}Y_2}$. Otherwise, an error event $E_{2,[l]}$ will be declared. 

Suppose the first part of the decoding process is successful. At the second step,  the decoder calculates $\mathbf{X}^n_{11,[l-1]}$, and $\mathbf{X}^n_{21,[l-1]}$. This is possible, because  $\mathbf{X}^n_{11,[l-1]}$, and $\mathbf{X}^n_{21,[l-1]}$ are functions of the messages.  The decoder, then, subtracts $\mathbf{X}^n_{11,[l-1]}\oplus \mathbf{X}^n_{21,[l-1]}$ from $Y_{1,[l-1]}$. The resulting vector is 
\begin{align*}
\tilde{\mathbf{Y}}^n=\mathbf{X}^n_{31,[l-1]}\oplus \tilde{N}^n_\delta.
\end{align*}
In this situation, the channel from $X_{31}$ to $\tilde{Y}$ is a binary additive channel with $\delta$ as the bias of the noise. At the third step, the decoder uses $\tilde{\mathbf{Y}}^n$ to decode the message of the third user, i.e., $\mathbf{W}^k_{3,[l-1]}$. In particular, the decoder finds unique $\mathbf{\tilde{w}}_3^k \in \FF_2^k$ such that $(\mathbf{\tilde{w}}_3^k\mathbf{G},  \mathbf{\tilde{Y}}^n)$ are jointly $\epsilon$-typical with respect to $P_{X_{31}\tilde{Y}}$. Otherwise, an error event $E_{3,[l]}$ is declared.

\textbf{Error Analysis:} 
We can show that this problem is equivalent to a point-to-point channel coding problem, where the channel is described by $Z=X\oplus \tilde{N}_\delta$. The average probability of error approaches zero, if $\frac{k}{n}\leq 1-h(\delta)$. 

Suppose there is no error in the decoding process of the third user. That is $E_{1,[l]}^c$ occurs. Therefore, $\mathbf{X}^n_{32,[l]}=\mathbf{X}^n_{22,[l]}\oplus \mathbf{X}^n_{12,[l]}$ with probability one. As a result, the channel in Fig. \ref{fig: Exp. second chann} is in the first state. This implies that the corresponding channel consists of two parallel binary additive channel with independent noises and bias $\delta$. Similar to the argument for $E_1$, it can be shown that $P(E_{2,[l]}| E_{1,[l]})\rightarrow 0$, if $\frac{k}{n}\leq 1-h(\delta)$. Lastly, we can show that conditioned on $E_{1,[l]}^c$ and $E_{2,[l]}^c$,  the probability of $E_{3,[l]}$ approaches zero, if $\frac{k}{n}\leq 1-h(\delta)$. 

As a result of the above argument, the average probability of error approaches $0$, if $\frac{k}{n}\leq 1-h(\delta)$. This implies that the rates $R_i=1-h(\delta), i=1,2,3$ are achievable, and the proof is completed.
\end{IEEEproof}

\section{Proof of Lemma \ref{lem: x_2+x_1 needs to be decoded}}\label{sec: proof of lemma 2 of the example}
\begin{IEEEproof}
Let $R_i$ be the rate of the $i$th encoder. We have $R_i\geq 1-h(\delta)-\epsilon$. We apply the generalized Fano's inequality (Lemma 4.3 in \cite{Kramer-thesis}) for decoding of the messages. More precisely, as $\bar{P}\leq \epsilon$, we have $$\frac{1}{M_1M_2M_3}H(\Theta_1, \Theta_2, \Theta_3| \mathbf{Y}^N)\leq h(\bar{P}) \leq h(\epsilon)$$

By the definition of the rate we have 
\begin{align}\nonumber
R_1+R_2+R_3&=\frac{1}{N}H(\Theta_1, \Theta_2, \Theta_3)\\\nonumber
& \leq \frac{1}{N}I(\Theta_1, \Theta_2, \Theta_3; \mathbf{Y}^n)+o(\epsilon)\\\nonumber
&\stackrel{(a)}{\leq } \frac{1}{N}I(\mathbf{X}^n_1, \mathbf{X}^n_2, \mathbf{X}^n_3; \mathbf{Y}^N)+o(\epsilon)\\\label{eq: last bound}
&\stackrel{(b)}{\leq } 3-\frac{1}{N}H(\mathbf{Y}^n|\mathbf{X}^n)+o(\epsilon),
\end{align}
where  $(a)$ is because of (\ref{eq: chann probabilities}), and for $(b)$ we use the fact that $Y$ is a vector of three binary random variables, which implies$\frac{1}{N}H(Y^N)\leq 3$.
As the channel is memoryless, and since (\ref{eq: chann probabilities}) holds, we have 
\begin{align*}
\frac{1}{N}H(\mathbf{Y}^n|\mathbf{X}^n)=\frac{1}{N}\sum_{l=1}^N H(Y_l | X_{1,l} X_{2,l} X_{3,l}).
\end{align*}
Let $P(X_{32,l}\neq X_{12,l}\oplus X_{12,l})=q_l$, for $l\in [1:N]$. Denote $\bar{q}_l=1-q_l$. We can show that,
\begin{align*}
H(Y_l | X_{1,l} X_{2,l} X_{3,l})=(1+2\bar{q}_l)h(\delta)+2q_l.
\end{align*}
We use the above argument, and the last inequality in (\ref{eq: last bound}) to give the following bound
\begin{align*}\nonumber
R_1+R_2+R_3 &\leq 3-\frac{1}{N}\sum_{l=1}^N [(1+2\bar{q}_l)h(\delta)+2q_l]+o(\epsilon)\\\label{eq: upper bound on I_N}
&= 3- 3 h(\delta)+\frac{1}{N}2(1-h(\delta))\sum_{l=1}^N q_l+o(\epsilon)
\end{align*}
By assumption $R_1+R_2+R_3  \geq 3(1-h(\delta)-\epsilon).$ Therefore, using the above bound we obtain,
\begin{align*}
\frac{3 \epsilon+o(\epsilon)}{2(1-h(\delta))} &\geq \frac{1}{N}\sum_{l=1}^N q_l \stackrel{(a)}{\geq }\frac{1}{N}\sum_{l\in \mathcal{I}^N_c} q_l,
\end{align*}
where $(a)$ holds, because we remove the summation over all $l \notin \mathcal{I}^N_c$. We defined  $\mathcal{I}^N_c$ as in the statement of this Lemma. Note that if $l \in \mathcal{I}^N_c$, then $q_l\geq c$. Finally, we obtain 
\begin{align*}
\frac{|\mathcal{I}^N_c|}{N} \leq \frac{3 \epsilon+o(\epsilon)}{2 c (1-h(\delta))}
\end{align*}
\end{IEEEproof}

\input{proof_lem_3.tex}
\input{proof_thm_1.tex}

\input{references.tex}
\end{document}

%% file: MATH_template.tex
\usepackage{tikz}
\usepackage[utf8]{inputenc}

\usepackage{mathrsfs}
\usepackage{amsmath} \usepackage{amsthm} \usepackage{amsfonts} \usepackage{amssymb} 
\usepackage{epstopdf}
\usepackage{graphicx}
\input{xypic}
\usepackage{bbm}

\usepackage{enumerate}

\newtheorem{theorem}{Theorem}
\newtheorem{preposition}{Proposition}
\newtheorem{lem}{Lemma}

\newtheorem{claim}{Claim}
\newtheorem{definition}{Definition}

\theoremstyle{definition}
\newtheorem{example}{Example}

\newtheorem*{prob*}{Problem}

\theoremstyle{remark}
\newtheorem{remark}{Remark}

\global\long\def\FF{\mathbb{F}}

%% file: proof_lem_3.tex
\section{Proof of Lemma \ref{lem: structure in the code}}\label{seq: lem 3}
\begin{proof}
Let $\mathcal{I}_c^N$ be as in Lemma \ref{lem: x_2+x_1 needs to be decoded}. The average probability of error for decoding $X_{12}^N\oplus X_{22}^N$ is bounded as 
\begin{align*}
\bar{P}_e&=\frac{1}{N}\sum_{l=1}^N P(X_{32,l}\neq X_{12,l}\oplus X_{22,l})\\
&=\frac{1}{N}\sum_{l\in \mathcal{I}_c^N} P(X_{32,l}\neq X_{12,l}\oplus X_{22,l})+\frac{1}{N}\sum_{l\notin \mathcal{L}_c^N} P(X_{32,l}\neq X_{12,l}\oplus X_{22,l})\\
&\leq \frac{|\mathcal{I}_c^N|}{N}+c(1-\frac{|\mathcal{I}_c^N|}{N})\\
&=(1-c)\frac{|\mathcal{I}_c^N|}{N} +c\\
&\leq (1-c)\frac{\eta(\epsilon)}{2c(1-h(\delta))}+c 
\end{align*}
As a result as $\epsilon\rightarrow 0$, then $\bar{P}_e\rightarrow c$. Since $c>0$ is arbitrary, $\bar{P}_e$ can be made arbitrary small. Hence, for any $\epsilon'>0$, and there exist $\epsilon>0$ and large enough $N$ such that $\bar{P}_e < \epsilon'$. Note that $X^N_{32}$ is a function of $M_3, Y_1^N, Y_{12}^N$ and $Y_{22}^N$. Next we argue that to get $\bar{P}_e < \epsilon'$, it is enough for $X_{32}^N$ to be a function of $M_3, Y_1^N$.  More precisely, given $X_{32, l}$, the random variables $Y_{12,l}$ and $Y_{22,l}$ are independent of $X_{12, l}\oplus X_{22, l}$. To see this, we need to consider two cases.  If $X_{32, l}=X_{12, l}\oplus X_{22, l}$ then the argument follows trivially. Otherwise, $Y_{12,l}=X_{12,l}\oplus N_{1/2}$, where $N_{1/2}\sim Ber(1/2)$, and it is independent of $X_{12,l}$. Hence in this case,  $Y_{12,l}$ is independent of $X_{12,l}$. Similarly, $Y_{22,l}$ is independent of $X_{22,l}$. 

By subtracting $X_{31}^N$ from $Y_1^N$, we get $Z^N := X_{11}^N \oplus X_{21}^N\oplus N_{\delta}^N$. Next, we argue that the third encoder uses $Z^N$ to decode $X_{12}^N\oplus X_{22}^N$. Since $M_3$ is independent of $M_1$ and $M_2$, it is independent of $X_{1j}^N, X_{j2}^N$ for $j=1,2$. Therefore $Z^N$ is independent of $M_3$. Hence, $X_{32}^N$ is function of $Z^N$. Intuitively, we convert the problem of decoding $X_{11}^N \oplus X_{21}^N$ to a point to point channel coding problem. The channel in this case is a binary additive channel with noise $N_\delta \sim Ber(\delta)$. In this channel coding problem the codebook at the encoder is $\mathcal{C}_{12}\oplus \mathcal{C}_{22}$.  The capacity of this channel equals $1-h_b(\delta)$. Since the average probability of error is small,  we can use the generalized Fano's inequality to bound the rate of the encoder. As a result, it can be shown that  
\begin{align}
\frac{1}{N}\log_2||\mathcal{C}_{12}\oplus \mathcal{C}_{22}|| \leq 1-h_b(\delta)+ \eta(\epsilon),
\end{align}
where $\eta(\epsilon)\rightarrow 0$ as $\epsilon \rightarrow 0$. 
\begin{claim}
The following bound holds 
\begin{align} \label{eqe: bound on C_12 and C_22}
\frac{1}{N}\log_2||\mathcal{C}_{j2}|| \geq 1-h_b(\delta)- \gamma_j(\epsilon),
\end{align}

 where $j=1,2$ and $\gamma_j(\epsilon)\rightarrow 0$ as $\epsilon \rightarrow 0$.
\end{claim}
\begin{proof}[Outline of the proof]
First, we  show that the decoder must decode $M_3$ from $Y_1^N$. We argued in the above that $X_{32}^N$ is independent of $M_3$. Hence, the message $M_3$ is encoded only to $X_{31}^N$. Since $X_{31}^N$ is sent though the first channel in Example 1, the decoder must decode $M_3$ from  $Y_1^N$. Next, we argue that the receiver must decode $M_1$ and $M_2$ from $Y_{21}^N$ and $Y_{22}^N$, respectively. Note that the rate of the third encoder is $1-h_b(\delta)$, which equals to the capacity of the first channel given $X_{11}^N \oplus X_{21}^N$. Therefore, the decoder can decode $M_3$, if it has $X_{11}^N \oplus X_{21}^N$. Hence, the decoder must reconstruct $X_{11}^N \oplus X_{21}^N$ from the second channel. It can be shown that this is possible, if the decoder can decode $M_1 $ and $M_2$ from the second channel. As a result, from Fano's inequality, the bounds in the Claim hold.  
\end{proof}

Finally, using (7) and (\ref{eqe: bound on C_12 and C_22}) we get 
\begin{align*}
0 \leq \frac{1}{N}\log_2||\mathcal{C}_{12}\oplus \mathcal{C}_{22}||-\frac{1}{N}\log_2||\mathcal{C}_{j2}|| \leq \eta(\epsilon)+\gamma_j(\epsilon), \quad j=1,2.
\end{align*}
This completes the proof.
\end{proof}

%% file: proof_thm_1.tex
\section{Proof of Theorem \ref{thm: MAC-FB achievable}}\label{sec: thm 1}
\begin{proof}
We build upon QLCs and propose a new coding scheme. 
Let $W_i$ be a random variable with distribution $P_{W_i}$. Fix integer $k$ and $n$. Consider the set of all $\epsilon$-typical sequences $W_i^k$. Without loss of generality assume that the new message at the $i$th encoder is a sequence $w_i^k$ which is selected randomly and uniformly from $A_\epsilon^{(k)}(W_i)$.  In this case $M_i=|A_\epsilon^{(k)}(W_i)|$. 

Define $\mathcal{L}[l-2]$ as the list of highly likely messages corresponding to the block $l-2$ at the decoder. This list is defined as 
\begin{align*}
\mathcal{L}[l-2] \triangleq \{(\hat{w}_1,\hat{w}_2,\hat{w}_3) \in A_\epsilon^{(n)}(W_1,W_2,W_3): (Y_{[l-2]}, U_{[l-2]}, S_{1,[l-2]},S_{2,[l-2]},S_{3,[l-2]}) \in A_\epsilon^{(n)}(\tilde{Y},\tilde{U}, \tilde{S}_1,\tilde{S}_2,\tilde{S}_3)\}
\end{align*}  

\textbf{Codebook Construction:}
For each $1\leq l \leq L$ generate $M_{0,[l]}$ sequences $U_{[l,m]}$, each according to $P_U^n$, where $1\leq m\leq M_{0,[l]}$. For any vector $w_i^k \in \FF_2^k$, denote $$t_i(w_i^k)\triangleq w_i^k \mathbf{G}+b_i^n, \quad i=1,2,3,$$
where $\mathbf{G}$ is a $k\times n$ matrix with elements chosen randomly and uniformly from $\FF_2$, and $b_i^n $ is a vector selected randomly and uniformly from $\FF_2^n$. 
 
For each $u^n \in \mathcal{U}^n$ and $t^n ,v^n \in \FF_2^n$ generate $M_i$ sequences $X_{i,[l,m]}^n$ randomly with conditional distribution $\prod_{j=1}^n P(\cdot | u_j, t_j, v_j)$, where $m\in [1:M_i]$. Denote such sequences by $x_i(u^n, t^n, v^n, m_i)$.

\textbf{Initialization:}
For block $l=0$, set $M_{0,[0]}=1,U_{[0,1]}=0$ and .  For block $l=1$, set $M_{0,[1]}=1,U_{[1,1]}=0,\mathbf{v}_{i,[1]}=0$.

\textbf{Encoding}
\paragraph{\textbf{Block $l=1$}}
At block $l=1$, given a message $\mathbf{w}_{i, [1]} \in A_\epsilon^{(k)}(W_i)$, the $i$th encoder calculates $t_i(\mathbf{w}_{i,[1]})$. This sequence is denoted by $t_{i,[1]}$. Next the encoder $i$ calculates $x_i(u_{[0,1]},t_{i,[1]}, v_{i,[1]}, \mathbf{w}_{i,[1]})$. Denote such sequence by $\mathbf{x}_{i, [1]}$. Finally,  the $i's$ encoder sends $\mathbf{x}_{i, [1]}$. 

\paragraph{\textbf{Block $l=2$}}
At the beginning of the block $l=2$, each encoder $i$ receives $Y_{[1]}$ as a feedback from the channel. The encoder $i$ wishes to decode sum of the messages of the other two encoders. The first encoder finds unique $\hat{w}_{23} \in A_\epsilon^{(k)}(W_2+W_3)$ such that $$(\hat{w}_{23}\mathbf{G}+b_2+b_3, Y_{[0]})\in A_{\epsilon}^{(n)}(T_2+T_3, Y|u_{[0]} t_{1,[0]}, x_{1,[l]}).$$ 
Otherwise an encoding error will be declared. If $\hat{w}_{23}$ was unique, the encoder sets $v_{1,[2]}=\hat{w}_{23}\mathbf{G}+b_2+b_3$.
Similarly encoder $2$ finds unique $\hat{w}_{13}$ and determines $v_{2,[2]}$. Also encoder $3$ finds unique $\hat{w}_{12}$, and determines $v_{3,[2]}$. 

\paragraph{\textbf{Block $l>2$}}
At the beginning of the block $l>2$, each encoder $i$ receives $Y_{[l-1]}$ as a feedback from the channel. The encoder $i$ wishes to decode sum of the messages of the other two encoders from block $l-1$. Next, given $Y_{[l-2]}$, the encoder $i$ decodes the messages of the other two encoders from block $l-2$. 

The first decoding process is the same as the decoding process in block $l=2$. Suppose $\hat{w}_{jk}$ and $v_{i,[l]}$ are the outputs of this decoding process at the encoder $i$.  The next stage of the decoding process is as follows. The first encoder finds unique $\hat{w}_{2,[l-2]}\in A_\epsilon^{(k)}(W_2)$ and $\hat{w}_{3,[l-2]}\in A_\epsilon^{(k)}(W_3)$ such that

1) $\hat{w}_{2,[l-2]}+\hat{w}_{3,[l-2]}=\hat{w}_{23}$.

2) \begin{align*}
 \Big(& t_2(\hat{w}_{2,[l-2]}),x_2\big(u^n,t_2(\hat{w}_{2,[l-2]}), v_{2,[l-2]},\hat{w}_{2,[l-2]}\big),\\
  & t_3(\hat{w}_{3,[l-2]}),x_3\big(u^n,t_3(\hat{w}_{3,[l-2]}), v_{3,[l-2]},\hat{w}_{3,[l-2]}\big), Y_{[l-2]} \Big) \in A_\epsilon^{(n)}(\tilde{T}_2\tilde{X}_2\tilde{T}_3\tilde{X}_3\tilde{Y} |  s_{1,[l-2]} v_{2,[l-2]}, v_{3,[l-2]})
\end{align*} 

3) $ (\hat{v}_{2,[l-1]},\hat{v}_{2,[l-1]}, Y_{[l-1]})\in A_\epsilon^{(n)}(V_2V_3Y | u_{[l-1]} s_{1[l-1])}) $,

where $v_{i,[l-2]}$ is known at the encoder from the previous blocks, and $\hat{v}_{2,[l-1]}, \hat{v}_{3,[l-1]}$ are defined as 
\begin{align*}
&\hat{v}_{2,[l-1]}=(w_{1,[l-2]}+\hat{w}_{3,[l-2]})\mathbf{G}+b_1+b_3\\
&\hat{v}_{3,[l-1]}=(w_{1,[l-2]}+\hat{w}_{2,[l-2]})\mathbf{G}+b_1+b_2.
\end{align*}
If the messages are not unique, an error will be declared. 

The next step, the encoder creates the list $\mathcal{L}[l-2]$ as defined in the above. If $(w_{1,[l-2]},\hat{w}_{2,[l-2]},\hat{w}_{3,[l-2]}) \in \mathcal{L}[l-2]$, then the first encoder finds the index $m$ corresponding to $(w_{1,[l-2]},\hat{w}_{2,[l-2]},\hat{w}_{3,[l-2]})$. Then the encoder calculates the corresponding $u_{[l-2, m}$. Denote such sequence by $u_{[l]}$. This sequence is used for transmission of new messages at block $l$. If the decoding processes are successful, then the sequences $v_{1,[l]}$ and $u_{[l]}$ are determined. The next step is the encoding process, which is the same as in the block $l=1$.

\paragraph{\textbf{Decoding at block $l$}}
The decoder knows the list of highly likely messages . This list is $\mathcal{L}[l-2]$ as defined in the above. Given $Y_{[l]}$ the decoder wishes to decode $U_{[l]}$. Note that $U_{[l]}$ determines the index of the messages in $\mathcal{L}[l-2]$ which were transmitted at block $l-2$. This decoding process is performed by finding unique index $m\i [1: M_{0,[l]}]$ such that $$(U_{[l,m]}, Y_{[l]}) \in A_\epsilon^{(n)}(U, Y| u_{[l-1]}, y_{[l-1]})$$

\paragraph{\textbf{Error Analysis}}
There are three types of decoding errors: 1) error in decoding sum of the messages of the other two encoders, i.e., $\hat{w}_{jk}$ is not unique at the encoder $i$. 2) error in the decoding of the individual messages of the other encoders, i.e., $\hat{w}_{j,[l]},\hat{w}_{k,[l]}$ are not unique at the encoder $i$. 3) error at the decoder, i.e. the index $m$ is not unique. Using standard arguments for each type of the errors we get the following bounds:

The probability of the first type of the errors approaches zero, if fro any distinct $i,j,k \in \{1,2,3\}$ the following bound holds:
\begin{align}\label{eq: bound 1}
\frac{k}{n}H(W_j+W_k) \leq I(T_j+T_k; Y|U T_k V_k X_k).
\end{align}

The probability of the second type of the errors approaches zero, if
\begin{align}\label{eq: bound 2}
\frac{k}{n}H(W_i|W_j+W_k)\leq I(\tilde{X}_i \tilde{X}_j; \tilde{Y}| \tilde{U} \tilde{S}_k \tilde{V}_1 \tilde{V}_2 \tilde{V}_3 )
\end{align}

Note that the third type of error occurs with high probability, if $|\mathcal{L}[l]| > 2^{nI(U;Y | \tilde{U}, \tilde{Y})}$. It can be shown that for sufficiently large $n$,
\begin{align*}
P\{|\mathcal{L}[l]| < 2^{n\max_{\mathcal{A}\subseteq \{1,2,3\}} F_{\mathcal{A}} +o(\epsilon)} \} > 1-\epsilon, 
\end{align*}
where $$F_{\mathcal{A}} \triangleq \frac{k}{n} H(W_{\mathcal{A}})-I(X_\mathcal{A}; Y| U S_{\mathcal{A}^c} \tilde{V}_1, \tilde{V}_2, \tilde{V}_3) $$
Therefore, the probability of third type of the errors approaches zero, if the following bounds hold:
\begin{align*}
F_{\mathcal{A}}\leq I(U;Y | \tilde{U}, \tilde{Y}),
\end{align*} 
Using the definition of $F_{\mathcal{A}}$ and the above bound, we can get the following bound:
\begin{align} \label{eq: bound 3}
\frac{k}{n} H(W_{\mathcal{A}}) \leq I(X_\mathcal{A}; Y| U S_{\mathcal{A}^c} \tilde{V}_1, \tilde{V}_2, \tilde{V}_3)+ I(U;Y | \tilde{U}, \tilde{Y})
\end{align} 

Note that the effective rate of our coding scheme is $R_i \triangleq \frac{1}{n}\log_2 M_i= \frac{k}{n}H(W_i)$ for $i=1,2,3$. Finally, it can be shown that using this equation and the bounds in \eqref{eq: bound 1}, \eqref{eq: bound 2}, and \eqref{eq: bound 3}, the following bounds are achievable
\begin{align*}
R_\mathcal{A}&\leq I(X_\mathcal{A};Y|US_{\mathcal{A}^c} \tilde{V}_1 \tilde{V}_2 \tilde{V}_3)+I(U;Y|\tilde{U}\tilde{Y})\\
R_i+R_j &\leq I(T_i\oplus T_j; Y| UT_k X_k \tilde{V}_1 \tilde{V}_2 \tilde{V}_3)\\&+I(\tilde{X}_i \tilde{X}_j;\tilde{Y}|\tilde{U} \tilde{S}_k \tilde{V}_1 \tilde{V}_2 \tilde{V}_3 V_k)\\&+I(\tilde{X}_i \tilde{X}_j;Y | \tilde{U} \tilde{S}_k \tilde{V}_1 \tilde{V}_2 \tilde{V}_3 U S_k \tilde{Y})\\
R_i+R_j &\leq \frac{H(W_i)+H(W_j)}{H(W_i\oplus W_j)} I(T_i\oplus T_j;Y|UT_k X_k).
 \end{align*}
\end{proof}